\documentclass[submission,copyright,creativecommons]{eptcs}
\providecommand{\event}{ICLP 2019}

\usepackage{graphicx}  
\usepackage[linesnumbered,ruled,vlined]{algorithm2e}
\usepackage{amssymb}
\usepackage{amsmath}

\usepackage{amsthm}

\usepackage{breakurl}             
\usepackage{underscore}           

\newtheorem{theorem}{Theorem}
\newtheorem{proposition}{Proposition}
\newtheorem{lemma}{Lemma}
\newtheorem{definition}{Definition}
\newtheorem{example}{Example}

\newcommand{\lpmln}{LP\textsuperscript{MLN}{ }}
\newcommand{\lpmlnend}{LP\textsuperscript{MLN}}
\newcommand{\aspwc}{ASP\textsuperscript{wc}{ }}
\newcommand{\aspwcend}{ASP\textsuperscript{wc}}
\newcommand{\lglred}[2]{\left(\overline{{#1}_{#2}}\right)^{#2}}

\title{On the Strong Equivalences of \lpmln Programs}

\author{
	Bin Wang \qquad Jun Shen \qquad Shutao Zhang \qquad Zhizheng Zhang
	\institute{School of Computer Science and Engineering
		\\Southeast University, Nanjing, China\\}
	\email{\{kse.wang, junshen, shutao\_zhang, seu\_zzz\}@seu.edu.cn}
}

\def\titlerunning{On the Strong Equivalences of \lpmln Programs}
\def\authorrunning{B. Wang, J. Shen, S. Zhang \& Z. Zhang}

\begin{document}
\maketitle


\begin{abstract}
By incorporating the methods of Answer Set Programming (ASP) and Markov Logic Networks (MLN), \lpmln becomes a powerful tool for non-monotonic, inconsistent and uncertain knowledge representation and reasoning. 
To facilitate the applications and extend the understandings of \lpmlnend, 
we investigate the strong equivalences between \lpmln programs in this paper, 
which is regarded as an important property in the field of logic programming. 
In the field of ASP, two programs $P$ and $Q$ are strongly equivalent, iff for any ASP program $R$, the programs $P \cup R$ and $Q \cup R$ have the same stable models. 
In other words, an ASP program can be replaced by one of its strong equivalent without considering its context, 
which helps us to simplify logic programs, enhance inference engines, 
construct human-friendly knowledge bases etc. 
Since \lpmln is a combination of ASP and MLN, the notions of strong equivalences in \lpmln is quite different from that in ASP. 
Firstly, we present the notions of p-strong and w-strong equivalences between \lpmln programs. 
Secondly, we present a characterization of the notions by generalizing the SE-model approach in ASP. 
Finally, we show the use of strong equivalences in simplifying \lpmln programs, and present a sufficient and necessary syntactic condition that guarantees the strong equivalence between a single \lpmln rule and the empty program.  
\end{abstract}



\section{Introduction}
\lpmln \cite{Lee2016Weighted}, a newly developed knowledge representation and reasoning language, is designed to handle non-monotonic and uncertain knowledge by combining the methods of Answer Set Programming (ASP) \cite{Brewka2011ASP,Gelfond1988theSM} and Markov Logic Networks (MLN) \cite{Richardson2006mln}. 
Specifically, an \lpmln program can be viewed as a weighted ASP program, 
where each ASP rule is assigned a weight denoting its certainty degree,  
and each weighted rule is allowed to be violated by a set of beliefs associated with the program. 
For example, \lpmln rule ``$w ~:~ \leftarrow a, ~b. $'' is a weighted constraint denoting facts $a$ and $b$ are contrary, $w$ is the weight of the constraint. 
In the view of ASP, the set $X=\{a, ~b\}$ is impossible to be a belief set of any ASP programs containing the constraint, 
while in the context of \lpmlnend, $X$ is a valid belief set. 
Since $X$ violates the constraint, the weight $-w$ is regarded as the certainty degree of $X$. 
It is easy to observe that the example can also be encoded by weak constraints in ASP. 
From this perspective, \lpmln can be viewed as an extension of ASP with weak constraints, that is, ASP with weak rules. 
Besides, several inference tasks are introduced to \lpmln such as computing marginal probability distribution of beliefs, computing most probable belief sets etc., 
which makes \lpmln suitable for knowledge reasoning in the context that contains uncertain and inconsistent data. 
For example, Eiter and Kaminski \cite{Eiter2016Exploiting} used \lpmln in the tasks of classifying visual objects, 
and some unpublished work tried to use \lpmln as the bridge between text and logical knowledge bases. 

Recent results on \lpmlnend aim to establish the relationships among \lpmln and other logic formalisms \cite{Balai2016realtionship,Lee2017lpmln}, 
develop \lpmln solvers \cite{Lee2017ComputingLpmln,Wang2017ParallelLpmln,Wu2018LPMLNModels}, 
acquire the weights of rules automatically \cite{Lee2018WeightLearning}, 
explore the properties of \lpmln \cite{Wang2018Splitting} etc. 
All these results lay the foundation for the problems solving via \lpmlnend, 
however, many theoretical problems of \lpmln are still unsolved, 
which prevents the wider applications of \lpmlnend. 
In this paper, we investigate the strong equivalences between \lpmln programs, 
which is regarded as an important property in the field of logic programming. 
For two ASP programs $P$ and $Q$, they are strongly equivalent, iff for any ASP program $R$, the programs $P \cup R$ and $Q \cup R$ have the same stable models \cite{Lifschitz2001Strongly}. 
In other words, an ASP program can be replaced by one of its strong equivalent without considering its context, 
which helps us to simplify logic programs, enhance inference engines, 
construct human-friendly knowledge bases etc. 
For example, an ASP rule such that its positive and negative body have common atoms is strongly equivalent to $\emptyset$ \cite{Inoue2004EqUpdate,Lin2005Discover,Osorio2001Equivalence}, 
therefore, such kinds of rules can be eliminated from any context, which leads to a more concise knowledge base and makes the reasoning easier. 
By investigating the strong equivalences in \lpmlnend, 
it is expected to improve the knowledge base constructing and knowledge reasoning in \lpmlnend, 
furthermore,  
help us to facilitate the applications and extend the understandings of \lpmlnend.

Our contributions are as follows. 
Firstly, we define the notions of strong equivalences in \lpmlnend, that is, the p-strong and w-strong equivalences. 
As we showed in above example, a stable model defined in \lpmln is associated with a certainty degree, 
therefore, the notions of strong equivalences in \lpmln are also relevant to the certainty degree. 
Secondly, we present a model-theoretical approach to characterizing the defined notions, which can be viewed as a generalization of the strong-equivalence models (SE-model) approach in ASP \cite{Turner2001SE}. 
Finally, we show the use of the strong equivalences in simplifying \lpmln programs, 
and present a sufficient and necessary syntactic condition that guarantees the strong equivalences between a single \lpmln rule and the empty program.


\section{Preliminaries}
In this section, we review the knowledge representation and reasoning language \lpmln presented in \cite{Lee2016Weighted}. 
An \lpmln program is a finite set of weighted rules $w:r$, 
where $w$ is the weight of rule $r$, 
and $r$ is an ASP rule of the form
\begin{equation}
\label{eq:asp-rule-form}
l_1 ~ \vee ~ ... ~\vee ~ l_k ~\leftarrow~ l_{k+1}, ..., ~l_m, ~not~ l_{m+1}, ...,~not ~ l_n.
\end{equation}
where $l$s are literals, 
$\vee$ is epistemic disjunction, 
and $not$ is default negation. 
The weight $w$ of an \lpmln rule is either a real number or a symbol ``$\alpha$'' denoting ``infinite weight'', 
and if $w$ is a real number, the rule is called \textit{soft}, otherwise, it is called \textit{hard}. 
For convenient description, we introduce some notations.
By $\overline{M}$ we denote the set of unweighted ASP counterpart of an \lpmln program $M$, 
i.e. $\overline{M} = \{r ~|~ w:r \in M \}$. 
For an ASP rule $r$ of the form \eqref{eq:asp-rule-form}, 
the literals occurred in head, positive body, and negative body of $r$ are denoted by $h(r) = \{l_i ~|~ 1 \leq i \leq k\}$, 
$b^+(r)=\{l_i ~|~ k+1 \leq i \leq m\}$, 
and $b^-(r)=\{l_i ~|~ m+1 \leq i \leq n\}$ respectively. 
Therefore, an ASP rule $r$ of the form (\ref{eq:asp-rule-form}) can also be abbreviated as ``$h(r) \leftarrow b^+(r), ~not ~ b^-(r).$".  
By $lit(r) = h(r) \cup b^+(r) \cup b^-(r)$ we denote the set of literals occurred in rule $r$, and by $lit(\Pi) = \bigcup_{r \in \Pi} lit(r)$ we denote the set of literals occurred in an ASP program $\Pi$.

An \lpmln program is called \textit{ground} if its rules contain no variables. 
Usually, a non-ground \lpmln program is considered as a shorthand for the corresponding ground program, 
therefore, we limited our attention to the strong equivalences between ground \lpmln programs in this paper. 
For a ground \lpmln program $M$, we use $W(M)$ to denote the weight degree of  $M$, i.e. $W(M) = exp\left(\sum_{w:r \in M } w\right)$.  
A ground \lpmln  rule $w:r$ is satisfied by a consistent set $X$ of ground literals, 
denoted by $X \models w:r$, if $X\models r$ by the notion of satisfiability in ASP. 
An \lpmln program $M$ is satisfied by $X$, denoted by $X \models M$, if $X$ satisfies all rules in $M$.
By $M_X$ we denote the \textit{\lpmln reduct} of an \lpmln program $M$ w.r.t. $X$, i.e. $M_X=\{w:r \in M ~|~ X \models w:r\}$. 
A consistent set $X$ of literals is a stable model of an ASP program $P$, if $X$ satisfies all rules in $P^X$ and $X$ is minimal in the sense of set inclusion, where $P^X$ is the Gelfond-Lifschitz reduct (GL-reduct) of $P$ w.r.t. $X$, i.e. $P^X = \{ h(r) \leftarrow b^+(r).  ~|~  r \in P \text{ and } b^-(r) \cap X = \emptyset  \}$.
The set $X$ is a \textit{stable model} of an \lpmln program $M$ if $X$ is a stable model of the ASP program $\overline{M_X}$. 
And by $SM(M)$ we denote the set of all stable models of an \lpmln program $M$. 
For a stable model $X$ of an \lpmln program $M$, 
the \textit{weight degree} $W(M,X)$ of $X$ w.r.t. $M$ is defined as $W(M_X)$, 
and the \textit{probability degree} $P(M,X)$ of $X$ w.r.t. $M$ is defined as
\begin{equation}
\label{eq:probability-sm}
P(M,X) = \lim\limits_{\alpha \rightarrow \infty} \frac{W(M,X)}{\Sigma_{X'\in SM(M)}W(M,X')}
\end{equation}
For a literal $l$, the \textit{probability degree} $P(M,l)$ of $l$ w.r.t. $M$ is defined as 
\begin{equation}
\label{eq:probability-lit}
P(M,l) = \sum_{l \in X, ~X \in SM(M)} P(M,X)
\end{equation}
A stable model $X$ of an \lpmln program $M$ is called a \textit{probabilistic stable model} of $M$ if $P(M,X) \neq 0$.
By $PSM(M)$ we denote the set of all probabilistic stable models of $M$. 
It is easy to check that $X$ is a probabilistic stable model of $M$, iff $X$ is stable model of $M$ that satisfies the most hard rules.
Based on above definitions, there are two kinds of main inference tasks for an \lpmln program $M$ \cite{Lee2017ComputingLpmln}:  
\begin{itemize}
	\item[-] Maximum A Posteriori (MAP) inference: compute the stable models with the highest weight or probability degree of the program $M$, i.e. the most probable stable model;
	\item[-] Marginal Probability Distribution (MPD) inference: compute the probability degrees of a set of literals w.r.t. the program $M$.
\end{itemize}


\section{Strong Equivalences for \lpmlnend}
In this section, we investigate the strong equivalences in \lpmlnend. 
Firstly, we define the notions of strong equivalences based on two different certainty degrees in \lpmlnend. 
Secondly, we present a model-theoretical approach to characterizing the notions. 
Finally, we present the relationships among these notions.

\subsection{Notions of Strong Equivalences}
The notion of strong equivalence is built on the notion of ordinary equivalence, 
in this section, we define two notions of ordinary equivalences between \lpmln programs, which is relevant to the weight and probability defined for stable models in \lpmlnend. 

\begin{definition}[w-ordinary equivalence]
\label{def:lpmln-ordinary-equivalence-w}
Two \lpmln programs $L$ and $M$ are w-ordinarily equivalent, denoted by $L \equiv_w M$, 
if their stable models coincide, and for each stable model $X$ of the programs, $W(L,X) = W(M,X)$. 
\end{definition}

\begin{definition}[p-ordinary equivalence]
\label{def:lpmln-ordinary-equivalence-p}
Two \lpmln programs $L$ and $M$ are p-ordinarily equivalent, denoted by $L \equiv_p M$, 
if their stable models coincide, and for each stable model $X$ of the programs,  $P(L, X) = P(M, X)$. 
\end{definition}

From Definition \ref{def:lpmln-ordinary-equivalence-w} and Definition \ref{def:lpmln-ordinary-equivalence-p}, 
it can be observed that both of the w-ordinary and p-ordinary equivalences can guarantee two \lpmln programs have the same MAP and MPD inference results, 
and the p-ordinary equivalence is a little weaker, i.e. if two \lpmln programs are p-ordinarily equivalent, then they are w-ordinarily equivalent, but the inverse dose not hold generally. 
Based on the definitions of ordinary equivalences, we can define two kinds of strong equivalences between \lpmln programs.  

\begin{definition}[strong equivalences for \lpmlnend]
\label{def:lpmln-strong-equivalences}
For two \lpmln programs $L$ and $M$, 
\begin{itemize}
	\item[-] they are w-strongly equivalent, denoted by $L \equiv_{s,w} M$, 
	if for any \lpmln program $N$, $L \cup N \equiv_{w} M \cup N$; 
	\item[-] they are p-strongly equivalent, denoted by $L \equiv_{s,p} M$, 
	if for any \lpmln program $N$, $L \cup N \equiv_{p} M \cup N$.
\end{itemize}
\end{definition}

The notions of w-strong and p-strong equivalences can guarantee the faithful replacement of an \lpmln program in any context. 
Here, we introduce a new notion of strong equivalence, semi-strong equivalence, that does not guarantee the faithful replacement, 
but helps us to simplify the characterizations of other strong equivalences. 

\begin{definition}[semi-strong equivalence]
\label{def:lpmln-strong-equivalence-semi}
Two \lpmln programs $L$ and $M$ are semi-strongly equivalent, denoted by $L \equiv_{s,s} M$, 
if for any \lpmln program $N$, the programs $L \cup N$ and $M \cup N$ have the same stable models.
\end{definition}

\subsection{Characterizations of Strong Equivalences}
In this section, we present the characterizations for w-strong and p-strong equivalences. 
From Definition \ref{def:lpmln-strong-equivalences} and Definition \ref{def:lpmln-strong-equivalence-semi}, 
the notions of w-strong and p-strong equivalences can be viewed as the strengthened semi-strong equivalence by introducing the certainty evaluations. 
Therefore, we present the characterization of semi-strong equivalence firstly, which severs as the basis of characterizing w-strong and p-strong equivalences.

\subsubsection{Characterizing Semi-Strong Equivalence}
Here, we characterize the semi-strong equivalence between \lpmln programs by generalizing the strong-equivalence models (SE-models) approach presented in \cite{Turner2001SE}. 
For the convenient description, we introduce following notions. 

\begin{definition}[SE-interpretation]
	A strong equivalence interpretation (SE-interpretation) is a pair of consistent sets of literals $(X, Y)$ such that $X \subseteq Y$. 
	An SE-interpretation $(X,Y)$ is called \textit{total} if $X = Y$, and \textit{non-total} if $X \subset Y$.
\end{definition}

\begin{definition}[SE-models for \lpmlnend]
	\label{def:lse-model}
	For an \lpmln program $M$, an SE-interpretation $(X,Y)$ is an SE-model of $M$, if $X \models M'$ and $Y \models M'$, where $M' = \lglred{M}{Y}$.
\end{definition}

In Definition \ref{def:lse-model}, $M'$ is an ASP program obtained from $M$ by a three-step transformation. 
In the first step, $M_Y$ is obtained from $M$ by removing all rules that cannot be satisfied by $Y$, which is the \lpmln reduct of $M$ w.r.t. $Y$. 
In the second step, $\overline{M_Y}$ is obtained by dropping weight of each rule in $M_Y$.
In the third step, $(\overline{M_Y})^Y$ is obtained by the GL-reduct. 
Clearly, an SE-model for the \lpmln program $M$ is an SE-model of a consistent unweighted subset of $M$ that is obtained by \lpmln reduct, 
which means the definition of SE-models for \lpmln programs is built on the definition of SE-models for ASP programs. 
In what follows, we use $LSE(M)$ to denote the set of all SE-models of an \lpmln program $M$. 

\begin{definition}
	\label{def:lse-weight-degree}
	For an \lpmln program $M$ and an SE-model $(X,Y)$ of $M$, the weight degree $W(M,(X,Y))$ of $(X,Y)$ w.r.t. the program $M$ is defined as 
	\begin{equation}
	\label{eq:lse-model-weight-degree}
	W(M,(X,Y)) = W(M_Y) = exp \left( \sum_{w:r \in M_Y} w \right)
	\end{equation}
\end{definition}

\begin{example}
	\label{ex:se-model}
	Consider an \lpmln program $L = \{\alpha :   a \vee b.  ~~ 1 : b \leftarrow  not ~a. \}$. For the set $X = \{a, b\}$, it is easy to check that $X \models L$, therefore, the \lpmln reduct $L_X$ is $L$ itself. 
	By the definitions of GL-reduct, $(\overline{L})^X = \{a \vee b.\}$, 
	therefore, both $S_1 = (\{a\}, X)$ and $S_2 = (\{b\}, X)$ are SE-models of $L$, 
	and $W(L, S_1) = W(L, S_2) = W(L) = e^{\alpha + 1}$.
\end{example}

Now, we show some useful properties of the SE-models for \lpmln programs. 
Proposition \ref{prop:XX-lse-model-LM} is an immediate result according to the definition of SE-models. 
\begin{proposition}
	\label{prop:XX-lse-model-LM}
	Let $M$ be an \lpmln program and $(X, Y)$ an SE-interpretation, 
	\begin{itemize}
		\item[-] if $X=Y$, then $(X,Y)$ is an SE-model of $M$; 
		\item[-] $(X,Y)$ is not an SE-model of $M$, iff $X \not\models \lglred{M}{Y}$.
	\end{itemize}
\end{proposition}

Proposition \ref{prop:x-sm-mn} shows the relationships between the SE-models and the stable models of an \lpmln program. 

\begin{proposition}
	\label{prop:x-sm-mn}
	For an \lpmln program $M$ and a total SE-model $(X,X)$ of $M$, 
	\begin{itemize}
		\item[-] there must be an \lpmln program $N$ such that $X$ is a stable model of $M \cup N$, for example, $N = \{w : a . ~|~ a \in X\}$; 
		\item[-] $X$ is a stable model of $M$, iff  $(X',X) \not\in LSE(M)$  for any proper subset $X'$ of $X$. 
	\end{itemize} 
\end{proposition}

Based on above results,  a characterization of semi-strong equivalence between \lpmln programs is presented in Lemma \ref{lem:lpmln-sm-strong-equiv}. 

\begin{lemma}
	\label{lem:lpmln-sm-strong-equiv}
	Let $L$ and $M$ be two \lpmln programs, they are semi-strongly equivalent, 
	iff they have the same SE-models, i.e. $LSE(L) = LSE(M)$.
\end{lemma}


\begin{proof}
The proof proceeds basically along the lines of the corresponding proof by Turner \cite{Turner2001SE}. 

For the if direction, 
suppose $LSE(L) = LSE(M)$, we need to prove that for any \lpmln program $N$, the programs $L \cup N$ and $M \cup N$ have the same stable models. 
We use proof by contradiction. 
Assume $Y$ is a set of literals such that $Y \in SM(L \cup N) - SM(M \cup N)$.
By the definition, we have $Y \models \lglred{(L \cup N)}{Y} = \lglred{L}{Y} \cup \lglred{N}{Y}$. 
By Proposition \ref{prop:XX-lse-model-LM}, we have $(Y,Y)$ is an SE-model of $L$. 
Hence, $(Y,Y)$ is also an SE-model of $M$. 
Then, we have $Y \models \lglred{M}{Y}$ and $Y \models \lglred{(M \cup N)}{Y}$. 
By the assumption $Y \not\in SM(M \cup N)$, there exists a consistent set $X$ of literals such that $X \models \lglred{(M \cup N)}{Y}$, then we have $X \models \lglred{M}{Y}$ and  $X \models \lglred{N}{Y}$, hence, $(X,Y)$ is an SE-model of $M$, which means $(X,Y)$ is also an SE-model of $L$. 
By the definition of stable model, $Y$ cannot be a stable model of $L \cup N$, which contradicts with the assumption $Y \in SM(L \cup N)$. 
Therefore, the programs $L \cup N$ and $M \cup N$ have the same stable models, and the if direction of Lemma \ref{lem:lpmln-sm-strong-equiv} is proven.

For the only-if direction, suppose  $SM(L \cup N) = SM(M \cup N)$, we need to prove that $LSE(L) = LSE(M)$. 
We use proof by contradiction. 
Assume $(X,Y)$ is an SE-interpretation such that $(X,Y) \in LSE(L) - LSE(M)$. 
By Proposition \ref{prop:XX-lse-model-LM}, we have $X \not\models \lglred{M}{Y}$. 
Let $N = \{1 : a. ~|~ a \in X\} \cup \{1 : a \leftarrow b. ~|~ a, b \in Y-X \}$. 
We have $\lglred{(M \cup N)}{Y} = \lglred{M}{Y} \cup \overline{N}$. 
Let $X'$ be a set of literals such that $X' \subseteq Y$ and $X' \models \lglred{M}{Y} \cup \overline{N}$. 
By the construction of $N$, we have $X \subseteq X'$. 
Since $X \not\models \lglred{M}{Y}$, we have $X \neq X'$. 
Hence, there must exist a literal $l \in Y-X$ such that $l \in X'$.
By the construction of $N$, we have $(Y-X) \subseteq X'$, which means $X' = Y$.
By the definition of stable models, $Y$ is a stable model of $M \cup N$, which means $Y$ should also be a stable model of $L \cup N$. 
By the definition of stable model, $(X,Y)$ cannot be an SE-model of $L$, which contradicts with the assumption $(X,Y) \in LSE(L)$. 
Therefore, $L$ and $M$ have the same SE-models, and the only-if direction of Lemma \ref{lem:lpmln-sm-strong-equiv} is proven. 
\end{proof}

\subsubsection{Characterizing W-Strong and P-Strong Equivalences}
Now we present a main result of the paper, that is, the characterizations of w-strong and p-strong equivalences. 
Based on Lemma \ref{lem:lpmln-sm-strong-equiv}, Lemma \ref{lem:lpmln-pse-sufficient} provides a sufficient condition to characterize the p-strong equivalence for \lpmln programs. 

\begin{lemma}
	\label{lem:lpmln-pse-sufficient}
	Two \lpmln programs $L$ and $M$ are p-strongly equivalent,  if $LSE(L) = LSE(M)$, and there exist two constants $c$ and $k$ such that 
	for each SE-model $(X,Y) \in LSE(L)$, $W(L,(X,Y)) = exp(c + k*\alpha) * W(M,(X,Y))$. 
\end{lemma}

\begin{proof}

%

For two \lpmln programs $L$ and $M$, by Lemma \ref{lem:lpmln-sm-strong-equiv}, if $LSE(L) = LSE(M)$, $L$ and $M$ are semi-strongly equivalent, i.e. for any \lpmln program $N$, $SM(L \cup N) = SM(M \cup N)$.  
Suppose there exist two constants $c$ and $k$ such that for each SE-model $(X,Y) \in LSE(L)$, $W(L,(X,Y)) = exp(c + k*\alpha) * W(M,(X,Y))$, we need to show that $L$ and $M$ are p-strongly equivalent. 
Let $N$ be an \lpmln program, it is easy to check that $X$ is a probabilistic stable model of $L \cup N$ iff $X$ is a probabilistic stable model of $M \cup N$, i.e. $PSM(L \cup N) = PSM(M \cup N)$. 
For a stable model $X \in PSM(L \cup N)$, the probability degree of $X$ can be reformulated as 
\begin{equation}
\label{eq:p-se-probability-equal}
\begin{split}
P(L \cup N, X)
& = \frac{W(L \cup N,X)}{\Sigma_{X'\in PSM(L \cup N)}W(L \cup N,X')} 
 = \frac{ exp(c + k*\alpha) * W(M \cup N,X)}{exp(c + k*\alpha) * \Sigma_{X'\in PSM(M \cup N)}  W(M \cup N,X')} \\
& = \frac{W(M \cup N,X)}{\Sigma_{X'\in PSM(M \cup N)} W(M \cup N,X')} 
= P(M \cup N, X)
\end{split} 
\end{equation}
By the definition of p-strong equivalence, we have  $L \equiv_{s,p} M$.
\end{proof}

The condition in Lemma \ref{lem:lpmln-pse-sufficient}, called PSE-condition, is sufficient to characterize the p-strong equivalence. 
One may ask that whether the PSE-condition is also necessary. 
To answer the question, we need to consider the hard rules of \lpmln particularly. 
For \lpmln programs containing no hard rules, it is easy to check that the PSE-condition is necessary. 
But for arbitrary \lpmln programs, this is not an immediate result, which is shown as follows. 
Firstly, we introduce some notations. 
For a set $U$ of literals, we use $2^U$ to denote the power set of $U$, 
and use $2^{U^+}$ to denote the maximal consistent part of the power set of $U$, 
i.e. $2^{U^+} = \{X \in 2^U~|~ X \text{ is consistent }\}$.

\begin{lemma}
	\label{lem:pse-onlyif-part1}
	For two p-strongly equivalent \lpmln programs $L$ and $M$, let $N_1$ and $N_2$ be arbitrary \lpmln programs such that $PSM(L \cup N_1) \cap PSM(L \cup N_2) \neq \emptyset$. 
	There exist two constants $c$ and $k$ such that for any SE-models $(X,Y)$ of $L$, 
	if $Y \in PSM(L \cup N_1) \cup PSM(L \cup N_2)$, then $W(L,(X,Y)) = exp(c + k*\alpha) * W(M,(X,Y))$.
\end{lemma}

By Lemma \ref{lem:pse-onlyif-part1}, for two p-strongly equivalent \lpmln programs $L$ and $M$, 
to prove the necessity of the PSE-condition, we need to find a set $E$ of \lpmln programs satisfying 

\begin{itemize}
	\item[-] $\forall N_1, ~N_2 \in E$, $PSM(L \cup N_1) \cap PSM(L \cup N_2) \neq \emptyset$; and 
	\item[-] $\bigcup_{N \in E} PSM(L \cup N) = 2^{U^+}$, where $U$ is the set of literals occurred in $L$ and $M$, i.e. $U = lit(\overline{L \cup M})$. 
\end{itemize}
Above set $E$ is called a set of necessary extensions w.r.t. \lpmln programs $L$ and $M$. 
As shown in Proposition \ref{prop:XX-lse-model-LM}, an arbitrary total SE-interpretation is an SE-model of an \lpmln program, 
therefore, 
if there exists a set of necessary extensions of two p-strongly equivalent programs $L$ and $M$, 
then the necessity of the PSE-condition can be proven. 
In what follows, we present a method to construct a set of necessary extensions. 

\begin{definition}
	\label{def:r-mxa}
	For two consistent sets $X$ and $Y$ of literals, and an atom $a$ such that $a \not\in X \cup Y$, 
	by $R(X,Y,a)$ we denote an \lpmln program as follows
	\begin{align}
	&\alpha ~:~ \leftarrow X, ~not ~ Y, ~a. \\
	&\alpha ~:~ a \leftarrow X, ~not ~ Y.
	\end{align}
\end{definition}

\begin{definition}[flattening extension]
	\label{def:flattening-extension}
	For an \lpmln program $M$ and a set $U$ of literals such that $lit(\overline{M}) \subseteq U$, 
	a flattening extension $T^k(M, U)$ of $M$ w.r.t. $U$ is defined as 
	\begin{itemize}
		\item[-] $T^0(M, U) = M \cup N_0$;
		\item[-] $T^{i+1}(M, U) = T^i(M, U) \cup R(X \cap U, U-X,a_{i+1})$, 
	\end{itemize}
	where $N_0$ is a set of weighted facts constructed from $U$, i.e.  $N_0 = \{\alpha : a_k. ~|~ a_k \in U\}$, $X$ is a probabilistic stable model of $T^i(M,U)$, i.e. $X \in PSM(T^i(M,U))$, and $a_{i+1} \not\in lit\left(\overline{T^i(M ,U)}\right)$. 
\end{definition}

According to the splitting set theorem of \lpmln \cite{Wang2018Splitting}, the  flattening extension has following properties. 
\begin{proposition}
	\label{prop:flattening-extension}
	For an \lpmln program $M$ and a set $U$ of literals, 
	if $T^{k+1}(M, U)$ is constructed from $T^k(M,U)$ by adding $R(X \cap U, U-X,a_{k+1})$, then we have 
	\begin{itemize}
		\item[-] $SM(T^0(M,U)) = 2^{U^+}$; 
		\item[-] $SM(T^{k+1}(M, U)) = SM(T^k(M, U)) \cup  
		\{ Y \cup \{a_{k+1}\} ~|~ Y \in SM(T^k(M, U)), \text{ and } Y \cap U = X \cap U\} \}$; and 
		\item[-] the weight degrees of stable models have following relationships
		\begin{equation}
		W(T^{k+1}(M, U), Y) = \begin{cases}
		W(T^{k}(M, U), Y) * e^{2\alpha} & \text{ if } Y \cap U \neq X \cap U, \\
		W(T^{k}(M, U), Y) * e^{\alpha} & \text{ otherwise. }
		\end{cases}
		\end{equation}
		and for two stable models $Y$ and $Z$ of $T^k(M, U)$, if $Y \cap U = Z \cap U$, then $W(T^k(M, U), Y) = W(T^k(M, U), Z)$.
	\end{itemize}
\end{proposition}

\begin{example}
	\label{ex:flattening-extension}
	Let $L$ be the \lpmln program in Example \ref{ex:se-model}, and a set of literals $U = \{a, b\}$. 
	By Definition \ref{def:flattening-extension}, $T^0(L, U) = L \cup \{\alpha : a. ~ \alpha : b.  \}$, 
	it is easy to check that all subsets of $U$ are the stable models of $T^0(L, U)$, $U$ is the unique probabilistic stable model. 
	By Definition \ref{def:r-mxa}, $R(U, \emptyset, c_{1})$ is as follows 
	\begin{eqnarray}
	& \alpha ~:~ \leftarrow a, ~b, ~c_1.\\
	& \alpha ~:~ c_1 \leftarrow a, ~b.
	\end{eqnarray}
	and we have  $T^1(L, U) = T^0(L, U) \cup R(U, \emptyset, c_1)$. 
	The stable models and their weight degrees of $L$, $T^0(L, U)$, and $T^1(L, U)$ are shown in Table \ref{tab:flattening-extension}.
	From the table, we can observe that the flattening extension can be used to adjust the sets of literals that satisfy the most hard rules. 
	\begin{table}
		\centering
		\caption{Computing Results in Example \ref{ex:flattening-extension}}
		\label{tab:flattening-extension}
		\begin{tabular}{cccccc}
			\hline
			Weight & $\emptyset$ & $\{a\}$  & $\{b\}$ & $\{a, b\}$ & $\{a, b, c_1\}$ \\
			\hline
			$L$ & $e^0$ & $e^{\alpha + 1}$ & $e^{\alpha + 1}$ & - & - \\
			$T^0(L, U)$ & $e^0$ & $e^{2\alpha + 1}$ & $e^{2\alpha + 1}$ & $e^{3\alpha + 1}$ & - \\
			$T^1(L, U)$ & $e^{2\alpha}$ & $e^{4\alpha + 1}$ & $e^{4\alpha + 1}$ & $e^{4\alpha + 1}$ &  $e^{4\alpha + 1}$ \\
			\hline
		\end{tabular}
	\end{table}
\end{example}

\begin{lemma}
	\label{lem:flattening-extension}
	Let $L$ and $M$ be two p-strongly equivalent \lpmln programs, and $U = lit(\overline{L \cup M})$. 
	For two consistent subsets $X$ and $Y$ of $U$, there exists a flattening extension $T^k(L, U)$ such that $X$ and $Y$ are probabilistic stable models of $T^k(L, U)$.
\end{lemma}

Lemma \ref{lem:flattening-extension} provides a method to construct a set of necessary extensions of two p-strongly equivalent \lpmln programs by constructing a set of flattening extensions, 
which means the PSE-condition is necessary to characterize the p-strong equivalence for \lpmln programs.

\begin{theorem}
	\label{thm:lpmln-strong-equivalence-w}
	Let $L$ and $M$ be two \lpmln programs, 
	\begin{itemize}
		\item[(i)] $L$ and $M$ are p-strongly equivalent iff $LSE(L) = LSE(M)$, and there exist two constants $c$ and $k$ such that 
		for each SE-model $(X,Y) \in LSE(L)$, $W(L,(X,Y)) = exp(c + k*\alpha) * W(M,(X,Y))$; 
		\item[(ii)] $L$ and $M$ are w-strongly equivalent iff they are p-strongly equivalent and the constants $c=k=0$.
	\end{itemize}
\end{theorem}

\begin{example}
	\label{ex:p-strong-equivalence}
	Consider \lpmln programs $L = \{~ w_1 : a \vee b. ~ w_2 : b \leftarrow a.  \}$ and $M = \{w_3: b. ~   w_4: a \leftarrow not ~b. \}$, 
	where $w_i ~(1 \leq i \leq 4)$ is a variable denoting the weight of corresponding rule. 
	It is easy to check that $(\{b\}, \{a, b\})$ is the unique non-total SE-model of $L$ and $M$, therefore, $L$ and $M$ are semi-strongly equivalent. 
	If the programs are also p-strongly equivalent, we have following system of linear equations, where $\mathcal{C} =  exp(k*\alpha + c)$ and $U = \{a, b\}$.
	
	\begin{equation}
	\left\{
	\begin{array}{l}
	W(L, (\emptyset, \emptyset)) = W(M, (\emptyset, \emptyset)) * \mathcal{C} \\
	W(L, (\{a\}, \{a\})) = W(M, (\{a\}, \{a\})) * \mathcal{C} \\
	W(L, (\{b\}, \{b\})) = W(M, (\{b\}, \{b\})) * \mathcal{C} \\
	W(L, (U, U)) = W(M, (U, U)) * \mathcal{C}
	\end{array}
	\right.
	\Rightarrow
	\left\{
	\begin{array}{l}
	w_2 = c + k*\alpha \\
	w_1 = w_4 + c + k * \alpha \\
	w_1 + w_2 = w_3 + w_4 + c + k * \alpha
	\end{array}
	\right.
	\end{equation}
	Solve the system of equations, 
	we have $L$ and $M$ are p-strongly equivalent iff $w_2 = w_3 = c + k*\alpha$ and $w_1 = w_4 + c + k*\alpha$; 
	and they are w-strongly equivalent iff $w_2 = w_3 = 0$ and $w_1 = w_4$. 
\end{example}


\section{Simplifying \lpmln Programs}
The notions of strong equivalences can be used to study the simplifications of logic programs. 
Specifically, if \lpmln program $L$ and $M$ are strongly equivalent, and the program $M$ is easier to solve or more friendly for human, then $L$ can be replaced by $M$.
In this section, we investigate the simplifications of \lpmln programs via using the notions of strong equivalences. 
In particular, we present an algorithm to simplify and solve \lpmln programs based on strong equivalences firstly.  
Then, we present some syntactic conditions that guarantee the strong equivalence between a single \lpmln rule and the empty set $\emptyset$, 
which can be used to check the strong equivalences efficiently. 

\begin{definition}
\label{def:valid-rule}
An \lpmln rule $w:r$ is called semi-valid, if $w:r$ is semi-strongly equivalent to $\emptyset$; 
the rule is called valid, if $w:r$ is p-strong equivalent to $\emptyset$.
\end{definition}

In Definition \ref{def:valid-rule}, 
we specify two kinds of \lpmln rules w.r.t semi-strong and p-strong equivalences. 
Obviously, a valid \lpmln rule can be eliminated from any \lpmln programs, 
while a semi-valid \lpmln rule cannot. 
By the definition,  eliminating a semi-valid \lpmln rule does not change the stable models of original programs, 
but changes the probability distributions of the  stable models, 
which means it may change the probabilistic stable models of original programs. 

\begin{example}
\label{ex:sm-valid-rules}
Consider three \lpmln programs $L = \{ \alpha : a \leftarrow a.  \}$, $M = \{ \alpha : \leftarrow a.  \}$, and $N = \{1 : a. \}$. 
It is easy to check that rules in $L$ and $M$ are valid and semi-valid, respectively. 
Table \ref{tab:sm-valid-rules} shows the stable models and their probability degrees of \lpmln programs $N$, $L \cup N$, and $M \cup N$. 
It can be observed that eliminating the rule of $M$ from $M \cup N$ makes all stable models of $M \cup N$ probabilistic, 
which means semi-valid rules cannot be eliminated directly.

\begin{table}
\centering
\caption{Computing Results in Example \ref{ex:sm-valid-rules}}
\label{tab:sm-valid-rules}
\begin{tabular}{cccc}
\hline
Stable Model $X$ & $~~P(N, X)~~$ & $~~P(L \cup N, X)~~$ & $~~P(M \cup N, X)~~$ \\
\hline
$\emptyset$ & $0.27$ & $0.27$ & $1$ \\
$\{a\}$ & $0.73$ & $0.73$ &  $0$ \\
\hline 
\end{tabular}
\end{table}
\end{example}

Algorithm \ref{alg:solver} provides a framework to simplify and solve \lpmln programs based on the notions of semi-valid and valid \lpmln rules. 
Firstly, simplify an \lpmln program $M$ by removing all semi-valid and valid rules (line 2 - 8). 
Then, compute the stable models of the simplified \lpmln program via using some existing \lpmln solvers, such as LPMLN2ASP, LPMLN2MLN \cite{Lee2017ComputingLpmln}, and LPMLN-Models \cite{Wu2018LPMLNModels} ect. 
Finally, compute the probability degrees of the stable models w.r.t. the simplified program and all semi-valid rules (line 9 - 12). 
The correctness of the algorithm can be proved by corresponding definitions.

\begin{algorithm}
\caption{Simplify and Solve \lpmln Programs}
\KwIn{an \lpmln program $M$}
\KwOut{stable models of $M$ and their probability degrees}
\label{alg:solver}
$S = \emptyset$, $M' = M$ \;
\ForEach{$w:r \in M$}{ 
	\eIf{$w:r$ is valid}{
		$M' = M' -\{w:r\}$ \;
	}{
		\If{$w:r$ is semi-valid}{
			$S = S \cup \{w:r\}$ \;
			$M' = M' -\{w:r\}$ \;
		}
	}
}

$SM(M) = SM(M') = call\text{-}lpmln\text{-}solver(M')$\;

\ForEach{$X \in SM(M)$}{
	$W'(M, X) = exp\left(\sum_{w:r \in M' \cup S \text{ and } X \models w:r} ~w \right)$ \;
}

Compute probability degrees for each stable model $X$ by Equation \eqref{eq:probability-sm} and $W'(M, X)$\;
\Return{$SM(M)$ and corresponding probability degrees}
\end{algorithm}

\begin{table}
\centering
\caption{Syntactic Conditions}
\label{tab:syntactic-conditions}
\begin{tabular}{ccc}
\hline
Name & Definition & Strong Equivalence \\
\hline
TAUT & $h(r) \cap b^+(r) \neq \emptyset$ & p, semi\\
CONTRA & $b^+(r) \cap b^-(r) \neq \emptyset$ & p, semi\\
CONSTR1 & $h(r) = \emptyset$ & semi \\
CONSTR2 & $h(r) \subseteq b^-(r)$ &  semi \\
CONSTR3 & $~~~~h(r) = \emptyset$, $b^+(r) = \emptyset$, and $b^-(r) = \emptyset~~~~$ & p, semi \\
\hline
\end{tabular}
\end{table}

In Algorithm \ref{alg:solver}, a crucial problem is to decide whether an \lpmln rule is valid or semi-valid. 
Theoretically, it can be done by checking the SE-models of a rule, 
however, the approach is highly complex in computation. 
Therefore, we investigate the syntactic conditions for the problem. 
Table \ref{tab:syntactic-conditions} shows five syntactic conditions for a rule $r$, 
where TAUT and CONTRA have been introduced to investigate the program simplification of ASP \cite{Osorio2001Equivalence,Eiter2004SimplifyingLP}, 
CONSTR1 means the rule $r$ is a constraint, and CONSTR3 is a special case of CONSTR1. 
Rules satisfying CONSTR2 is usually used to eliminate constraints in ASP, 
for example, rule ``$\leftarrow a. $'' is equivalent to rule ``$p \leftarrow a, ~not ~p.$'', if the atom $p$ does not occur in other rules. 
Based on these conditions, we present the characterization of semi-valid and valid \lpmln rules.

\begin{theorem}
\label{thm:semi-valid}
An \lpmln rule $w:r$ is semi-valid, iff the rule satisfies one of TAUT, CONTRA, CONSTR1 and CONSTR2. 
\end{theorem}

\begin{theorem}
\label{thm:valid}
An \lpmln rule $w:r$ is valid, iff one of following condition is satisfied 
\begin{itemize}
	\item[-] rule $w:r$ satisfies one of TAUT, CONTRA, and CONSTR3; or 
	\item[-] rule $w:r$ satisfies CONSTR1 or CONSTR2, and $w = 0$.
\end{itemize}
\end{theorem}

Theorem \ref{thm:semi-valid} and Theorem \ref{thm:valid} can be proven by Lemma \ref{lem:lpmln-sm-strong-equiv} and Theorem \ref{thm:lpmln-strong-equivalence-w}. 
It is worthy noting that conditions CONSTR1 and CONSTR2 means the only effect of constraints in \lpmln is to change the probability distribution of inference results, 
which can also be observed in Example \ref{ex:flattening-extension}. 
In this sense, the constraints in \lpmln can be regarded as the weak constraints in ASP, 
and Algorithm \ref{alg:solver} is similar to the algorithm of solving ASP containing weak constraints. 
In both of algorithms, stable models are computed by removing (weak) constraints, 
and the certainty evaluations of the stable models are computed by combining these constraints. 

Combining Theorem \ref{thm:semi-valid} and Theorem \ref{thm:valid}, 
Algorithm \ref{alg:solver} is an alternative approach to enhance \lpmln solvers. 
In addition, Theorem \ref{thm:semi-valid} and Theorem \ref{thm:valid} also contribute to the field of knowledge acquiring. 
On the one hand, although it is impossible that rules of the form TAUT, CONTRA, and CONSTR3 are constructed by a skillful knowledge engineer, 
these rules may be obtained from data via rule learning. 
Therefore, we can use TAUT, CONTRA, and CONSTR3 as the heuristic information to improve the results of rule learning. 
On the other hand, CONSTR1 and CONSTR2 imply a kind of methodology of problem modeling in \lpmlnend, 
that is,  we can encode objects and relations by \lpmln rules and facts, 
and adjust the certainty degrees of inference results by \lpmln constraints. 
In fact, this is the core idea of ASP with weak constraints, \lpmln is more flexible by contrast, 
since \lpmln provides weak facts and rules besides weak constraints.


\section{Conclusion and Future Work}
In this paper, we present four kinds of notions of strong equivalences between \lpmln programs by comparing the certainty degrees of stable models in different ways, i.e. semi-strong, w-strong and p-strong equivalences, 
where w-strong equivalence is the strongest notion, and semi-strong equivalence is the weakest notion. 
For each notion, we present a sufficient and necessary condition to characterize it, 
which can be viewed as a generalization of SE-model approach in ASP. 
After that, we present a sufficient and necessary condition that guarantees the strong equivalence between a single \lpmln rule and the empty set, 
and we present an algorithm to simplify and solve \lpmln programs by using the condition. 
The condition can also be used to improve the knowledge acquiring and increase the understanding of the methodology of problems modeling in \lpmlnend. 

As we showed in the paper, there is a close relationship between \lpmln and ASP, 
especially, the constraints in \lpmln can be regarded as the weak constraints in ASP. 
Concerning related work, the strong equivalence for ASP programs with weak constraints (abbreviated to \aspwcend) has been investigated  \cite{Eiter2007ASPbpa}. 
It is easy to observe that the strong equivalence and corresponding characterizations of \aspwc can be viewed as a special case of the p-strong equivalence in ASP. 

For the future, we plan to improve the equivalences checking in the paper, and use these technologies to enhance \lpmln solvers. 
And we also plan to extend the strong equivalence discovering method introduced in \cite{Lin2005Discover} to \lpmlnend, which would help us to decide strong equivalence via some syntactic conditions. 


\section{Acknowledgments}
We are grateful to the anonymous referees for their useful comments.
The work was supported by the National Key Research and Development Plan of China (Grant No.2017YFB1002801).

\bibliographystyle{eptcs}
\bibliography{lpmln-equivalence-iclp}

\begin{thebibliography}{10}
\providecommand{\bibitemdeclare}[2]{}
\providecommand{\surnamestart}{}
\providecommand{\surnameend}{}
\providecommand{\urlprefix}{Available at }
\providecommand{\url}[1]{\texttt{#1}}
\providecommand{\href}[2]{\texttt{#2}}
\providecommand{\urlalt}[2]{\href{#1}{#2}}
\providecommand{\doi}[1]{doi:\urlalt{http://dx.doi.org/#1}{#1}}
\providecommand{\bibinfo}[2]{#2}

\bibitemdeclare{inproceedings}{Balai2016realtionship}
\bibitem{Balai2016realtionship}
\bibinfo{author}{Evgenii \surnamestart Balai\surnameend} \&
  \bibinfo{author}{Michael \surnamestart Gelfond\surnameend}
  (\bibinfo{year}{2016}): \emph{\bibinfo{title}{{On the Relationship between
  P-log and LP\({\textsuperscript{MLN}}\)}}}.
\newblock In \bibinfo{editor}{Subbarao \surnamestart Kambhampati\surnameend},
  editor: {\sl \bibinfo{booktitle}{Proceedings of the 25th International Joint
  Conference on Artificial Intelligence}}, pp. \bibinfo{pages}{915--921}.

\bibitemdeclare{article}{Brewka2011ASP}
\bibitem{Brewka2011ASP}
\bibinfo{author}{Gerhard \surnamestart Brewka\surnameend},
  \bibinfo{author}{Thomas \surnamestart Eiter\surnameend} \&
  \bibinfo{author}{Miros{\l}aw \surnamestart Truszczy{\'{n}}ski\surnameend}
  (\bibinfo{year}{2011}): \emph{\bibinfo{title}{{Answer Set Programming at a
  Glance}}}.
\newblock {\sl \bibinfo{journal}{Communications of the ACM}}
  \bibinfo{volume}{54}(\bibinfo{number}{12}), pp. \bibinfo{pages}{92--103},
  \doi{10.1145/2043174.2043195}.

\bibitemdeclare{article}{Eiter2007ASPbpa}
\bibitem{Eiter2007ASPbpa}
\bibinfo{author}{Thomas \surnamestart Eiter\surnameend},
  \bibinfo{author}{Wolfgang \surnamestart Faber\surnameend},
  \bibinfo{author}{Michael \surnamestart Fink\surnameend} \&
  \bibinfo{author}{Stefan \surnamestart Woltran\surnameend}
  (\bibinfo{year}{2007}): \emph{\bibinfo{title}{{Complexity results for answer
  set programming with bounded predicate arities and implications}}}.
\newblock {\sl \bibinfo{journal}{Annals of Mathematics and Artificial
  Intelligence}} \bibinfo{volume}{51}(\bibinfo{number}{2-4}), pp.
  \bibinfo{pages}{123--165}, \doi{10.1007/s10472-008-9086-5}.

\bibitemdeclare{incollection}{Eiter2004SimplifyingLP}
\bibitem{Eiter2004SimplifyingLP}
\bibinfo{author}{Thomas \surnamestart Eiter\surnameend},
  \bibinfo{author}{Michael \surnamestart Fink\surnameend},
  \bibinfo{author}{Hans \surnamestart Tompits\surnameend} \&
  \bibinfo{author}{Stefan \surnamestart Woltran\surnameend}
  (\bibinfo{year}{2004}): \emph{\bibinfo{title}{{Simplifying Logic Programs
  Under Uniform and Strong Equivalence}}}.
\newblock In: {\sl \bibinfo{booktitle}{Proceedings of the 7th International
  Conference on Logic Programming and Nonmonotonic Reasoning}}, pp.
  \bibinfo{pages}{87--99}, \doi{10.1007/978-3-540-24609-1_10}.

\bibitemdeclare{inproceedings}{Eiter2016Exploiting}
\bibitem{Eiter2016Exploiting}
\bibinfo{author}{Thomas \surnamestart Eiter\surnameend} \&
  \bibinfo{author}{Tobias \surnamestart Kaminski\surnameend}
  (\bibinfo{year}{2016}): \emph{\bibinfo{title}{{Exploiting Contextual
  Knowledge for Hybrid Classification of Visual Objects}}}.
\newblock In \bibinfo{editor}{J{\"{u}}rgen \surnamestart Dix\surnameend},
  \bibinfo{editor}{Lu{\'{i}}s~Fari{\~{n}}as \surnamestart del Cerro\surnameend}
  \& \bibinfo{editor}{Ulrich \surnamestart Furbach\surnameend}, editors: {\sl
  \bibinfo{booktitle}{Proceedings of the 15th European Conference on Logics in
  Artificial Intelligence}}, {\sl \bibinfo{series}{Lecture Notes in Computer
  Science}} \bibinfo{volume}{10021}, \bibinfo{publisher}{Springer Berlin
  Heidelberg}, \bibinfo{address}{Berlin, Heidelberg}, pp.
  \bibinfo{pages}{223--239}, \doi{10.1007/978-3-319-48758-8_15}.

\bibitemdeclare{inproceedings}{Gelfond1988theSM}
\bibitem{Gelfond1988theSM}
\bibinfo{author}{Michael \surnamestart Gelfond\surnameend} \&
  \bibinfo{author}{Vladimir \surnamestart Lifschitz\surnameend}
  (\bibinfo{year}{1988}): \emph{\bibinfo{title}{{The Stable Model Semantics for
  Logic Programming}}}.
\newblock In \bibinfo{editor}{Robert~A. \surnamestart Kowalski\surnameend} \&
  \bibinfo{editor}{Kenneth~A. \surnamestart Bowen\surnameend}, editors: {\sl
  \bibinfo{booktitle}{Proceedings of the Fifth International Conference and
  Symposium on Logic Programming}}, \bibinfo{publisher}{MIT Press}, pp.
  \bibinfo{pages}{1070--1080}.

\bibitemdeclare{incollection}{Inoue2004EqUpdate}
\bibitem{Inoue2004EqUpdate}
\bibinfo{author}{Katsumi \surnamestart Inoue\surnameend} \&
  \bibinfo{author}{Chiaki \surnamestart Sakama\surnameend}
  (\bibinfo{year}{2004}): \emph{\bibinfo{title}{{Equivalence of Logic Programs
  Under Updates}}}.
\newblock In: {\sl \bibinfo{booktitle}{Proceedings of the 9th European Workshop
  on Logics in Artificial Intelligence}}, \bibinfo{volume}{3229}, pp.
  \bibinfo{pages}{174--186}, \doi{10.1007/978-3-540-30227-8_17}.

\bibitemdeclare{article}{Lee2017ComputingLpmln}
\bibitem{Lee2017ComputingLpmln}
\bibinfo{author}{Joohyung \surnamestart Lee\surnameend},
  \bibinfo{author}{Samidh \surnamestart Talsania\surnameend} \&
  \bibinfo{author}{Yi~\surnamestart Wang\surnameend} (\bibinfo{year}{2017}):
  \emph{\bibinfo{title}{{Computing LP MLN using ASP and MLN solvers}}}.
\newblock {\sl \bibinfo{journal}{Theory and Practice of Logic Programming}}
  \bibinfo{volume}{17}(\bibinfo{number}{5-6}), pp. \bibinfo{pages}{942--960},
  \doi{10.1017/S1471068417000400}.

\bibitemdeclare{inproceedings}{Lee2016Weighted}
\bibitem{Lee2016Weighted}
\bibinfo{author}{Joohyung \surnamestart Lee\surnameend} \&
  \bibinfo{author}{Yi~\surnamestart Wang\surnameend} (\bibinfo{year}{2016}):
  \emph{\bibinfo{title}{{Weighted Rules under the Stable Model Semantics}}}.
\newblock In \bibinfo{editor}{Chitta \surnamestart Baral\surnameend},
  \bibinfo{editor}{James~P. \surnamestart Delgrande\surnameend} \&
  \bibinfo{editor}{Frank \surnamestart Wolter\surnameend}, editors: {\sl
  \bibinfo{booktitle}{Proceedings of the Fifteenth International Conference on
  Principles of Knowledge Representation and Reasoning:}},
  \bibinfo{publisher}{AAAI Press}, pp. \bibinfo{pages}{145--154}.

\bibitemdeclare{inproceedings}{Lee2018WeightLearning}
\bibitem{Lee2018WeightLearning}
\bibinfo{author}{Joohyung \surnamestart Lee\surnameend} \&
  \bibinfo{author}{Yi~\surnamestart Wang\surnameend} (\bibinfo{year}{2018}):
  \emph{\bibinfo{title}{{Weight Learning in a Probabilistic Extension of Answer
  Set Programs}}}.
\newblock In: {\sl \bibinfo{booktitle}{Proceedings of the 16th International
  Conference on the Principles of Knowledge Representation and Reasoning}}, pp.
  \bibinfo{pages}{22--31}.

\bibitemdeclare{inproceedings}{Lee2017lpmln}
\bibitem{Lee2017lpmln}
\bibinfo{author}{Joohyung \surnamestart Lee\surnameend} \&
  \bibinfo{author}{Zhun \surnamestart Yang\surnameend} (\bibinfo{year}{2017}):
  \emph{\bibinfo{title}{{LP\({\textsuperscript{MLN}}\), Weak Constraints, and
  P-log}}}.
\newblock In \bibinfo{editor}{Satinder~P. \surnamestart Singh\surnameend} \&
  \bibinfo{editor}{Shaul \surnamestart Markovitch\surnameend}, editors: {\sl
  \bibinfo{booktitle}{Proceedings of the Thirty-First {AAAI} Conference on
  Artificial Intelligence}}, \bibinfo{publisher}{AAAI Press}, pp.
  \bibinfo{pages}{1170--1177}.

\bibitemdeclare{article}{Lifschitz2001Strongly}
\bibitem{Lifschitz2001Strongly}
\bibinfo{author}{Valdimir \surnamestart Lifschitz\surnameend},
  \bibinfo{author}{David \surnamestart Pearce\surnameend} \&
  \bibinfo{author}{Agust{\'{\i}}n \surnamestart Valverde\surnameend}
  (\bibinfo{year}{2001}): \emph{\bibinfo{title}{{Strongly equivalent logic
  programs}}}.
\newblock {\sl \bibinfo{journal}{ACM Transactions on Computational Logic}}
  \bibinfo{volume}{2}(\bibinfo{number}{4}), pp. \bibinfo{pages}{526--541},
  \doi{10.1145/383779.383783}.

\bibitemdeclare{article}{Lin2005Discover}
\bibitem{Lin2005Discover}
\bibinfo{author}{Fangzhen \surnamestart Lin\surnameend} \& \bibinfo{author}{Yin
  \surnamestart Chen\surnameend} (\bibinfo{year}{2007}):
  \emph{\bibinfo{title}{{Discovering Classes of Strongly Equivalent Logic
  Programs}}}.
\newblock {\sl \bibinfo{journal}{Journal of Artificial Intelligence Research}}
  \bibinfo{volume}{28}, pp. \bibinfo{pages}{431--451}, \doi{10.1613/jair.2131}.

\bibitemdeclare{inproceedings}{Osorio2001Equivalence}
\bibitem{Osorio2001Equivalence}
\bibinfo{author}{Mauricio \surnamestart Osorio\surnameend},
  \bibinfo{author}{Juan~Antonio \surnamestart Navarro\surnameend} \&
  \bibinfo{author}{Jos{\'{e}} \surnamestart Arrazola\surnameend}
  (\bibinfo{year}{2001}): \emph{\bibinfo{title}{{Equivalence in Answer Set
  Programming}}}.
\newblock In: {\sl \bibinfo{booktitle}{Proceedings of the 11th International
  Workshop on Logic Based Program Synthesis and Transformation,}}, pp.
  \bibinfo{pages}{57--75}, \doi{10.1007/3-540-45607-4_4}.

\bibitemdeclare{article}{Richardson2006mln}
\bibitem{Richardson2006mln}
\bibinfo{author}{Matthew \surnamestart Richardson\surnameend} \&
  \bibinfo{author}{Pedro \surnamestart Domingos\surnameend}
  (\bibinfo{year}{2006}): \emph{\bibinfo{title}{{Markov logic networks}}}.
\newblock {\sl \bibinfo{journal}{Machine Learning}}
  \bibinfo{volume}{62}(\bibinfo{number}{1-2}), pp. \bibinfo{pages}{107--136},
  \doi{10.1007/s10994-006-5833-1}.

\bibitemdeclare{inproceedings}{Turner2001SE}
\bibitem{Turner2001SE}
\bibinfo{author}{Hudson \surnamestart Turner\surnameend}
  (\bibinfo{year}{2001}): \emph{\bibinfo{title}{{Strong Equivalence for Logic
  Programs and Default Theories (Made Easy)}}}.
\newblock In: {\sl \bibinfo{booktitle}{Proceedings of the 6th International
  Conference on Logic Programming and Nonmonotonic Reasoning}}, pp.
  \bibinfo{pages}{81--92}, \doi{10.1007/3-540-45402-0_6}.

\bibitemdeclare{inproceedings}{Wang2017ParallelLpmln}
\bibitem{Wang2017ParallelLpmln}
\bibinfo{author}{Bin \surnamestart Wang\surnameend} \&
  \bibinfo{author}{Zhizheng \surnamestart Zhang\surnameend}
  (\bibinfo{year}{2017}): \emph{\bibinfo{title}{{A Parallel
  LP\({\textsuperscript{MLN}}\) Solver: Primary Report}}}.
\newblock In \bibinfo{editor}{Bart \surnamestart Bogaerts\surnameend} \&
  \bibinfo{editor}{Amelia \surnamestart Harrison\surnameend}, editors: {\sl
  \bibinfo{booktitle}{Proceedings of the 10th Workshop on Answer Set
  Programming and Other Computing Paradigms}}, \bibinfo{publisher}{CEUR-WS},
  \bibinfo{address}{Espoo, Finland}, pp. \bibinfo{pages}{1--14}.

\bibitemdeclare{inproceedings}{Wang2018Splitting}
\bibitem{Wang2018Splitting}
\bibinfo{author}{Bin \surnamestart Wang\surnameend}, \bibinfo{author}{Zhizheng
  \surnamestart Zhang\surnameend}, \bibinfo{author}{Hongxiang \surnamestart
  Xu\surnameend} \& \bibinfo{author}{Jun \surnamestart Shen\surnameend}
  (\bibinfo{year}{2018}): \emph{\bibinfo{title}{{Splitting an
  LP\({\textsuperscript{MLN}}\) Program}}}.
\newblock In: {\sl \bibinfo{booktitle}{Proceedings of the Thirty-Second {AAAI}
  Conference on Artificial Intelligence}}, pp. \bibinfo{pages}{1997--2004}.

\bibitemdeclare{inproceedings}{Wu2018LPMLNModels}
\bibitem{Wu2018LPMLNModels}
\bibinfo{author}{Wei \surnamestart Wu\surnameend}, \bibinfo{author}{Hongxiang
  \surnamestart Xu\surnameend}, \bibinfo{author}{Shutao \surnamestart
  Zhang\surnameend}, \bibinfo{author}{Jiaqi \surnamestart Duan\surnameend},
  \bibinfo{author}{Bin \surnamestart Wang\surnameend},
  \bibinfo{author}{Zhizheng \surnamestart Zhang\surnameend},
  \bibinfo{author}{Chenglong \surnamestart He\surnameend} \&
  \bibinfo{author}{Shiqiang \surnamestart Zong\surnameend}
  (\bibinfo{year}{2018}): \emph{\bibinfo{title}{{LPMLNModels: A Parallel Solver
  for LPMLN}}}.
\newblock In: {\sl \bibinfo{booktitle}{2018 IEEE 30th International Conference
  on Tools with Artificial Intelligence (ICTAI)}}, \bibinfo{publisher}{IEEE},
  pp. \bibinfo{pages}{794--799}, \doi{10.1109/ICTAI.2018.00124}.

\end{thebibliography}

\end{document}